\documentclass[12pt,letterpaper,reqno]{amsart}

\usepackage{amsmath}
\usepackage{amsfonts}
\usepackage{amssymb}
\usepackage{tikz}
\usepackage{multicol}
\usetikzlibrary{automata,positioning}
\usepackage{graphicx}
\usepackage{subfig}
\usepackage{array}
\usepackage[margin=.75in]{geometry}
\usepackage{enumerate}
\usepackage{mathrsfs}
\usepackage{color}

\theoremstyle{plain}

\newtheorem{prop}{Proposition}

\numberwithin{equation}{section}
\usepackage{pgfplots}
\usepackage{tikz}
\usepackage{fancyhdr}
\newcommand{\bcd}{\[\begin{tikzcd}}
\newcommand{\ecd}{\end{tikzcd}\]}

\renewcommand{\P}{\mathbb{P}}

\DeclareMathOperator{\st}{st}

\DeclareMathOperator{\PE}{PE}
\DeclareMathOperator{\DKL}{\text{D}_{\text{KL}}}
\renewcommand{\S}{\mathcal{S}}
\usepackage{hyperref}

\include{somethingsomethingdata}

\begin{document}
\title[Permutation Patterns In Time Series]{Random Walk Null Models for Time Series Data}
\author{Daryl DeFord \and Katherine Moore}
\address{Department of Mathematics\\
Dartmouth College}
\email{ddeford@math.dartmouth.edu \and
moorek@math.dartmouth.edu}

\date{\today}

\begin{abstract}
Permutation entropy has become a standard tool for time series analysis that exploits  the temporal properties of these data sets. Many current applications use an approach based on Shannon entropy, which implicitly assumes an underlying uniform distribution of patterns.  In this paper, we analyze {\em random walk null models } for time series  and determine the corresponding permutation distributions. These new techniques allow us to explicitly  describe the behavior of real--world data in terms of more complex generative processes. Additionally,  building on recent results of Martinez, we define a validation measure that allows us to determine  when a random walk is an appropriate model for a time series.  We demonstrate the usefulness of our methods using empirical data drawn from a variety of fields. \end{abstract}

\maketitle

\section{Introduction}

In the past fifteen years, measures of entropy defined in terms of the distribution of ordinal patterns have become important tool in the analysis of time series. These methods effectively make use of the temporal structure of this type of data  in ways that are both computationally efficient and simple to implement.  In addition, permutation entropy is invariant under scaling of the data, i.e. under non-linear monotonic transformations, adding to its wide applicability \cite{amigo_combinatorial_2008, bandt_autocorrelation_2014}.  These techniques have found application in in many fields including economics \cite{bariviera_revisiting_2013,hou_characterizing_2017, lim_rapid_2014, ortiz-cruz_efficiency_2012, zunino_efficiency_2012, zunino_forbidden_2009}, medicine \cite{li_multiscale_2010,li_predictability_2007, liu_multiscale_2017, nicolaou_detection_2012,  olofsen_permutation_2008}, and physics \cite{weck_permutation_2015,zunino_permutation_2008}, among others. Three recent surveys \cite{keller_permutation_2017, riedl_practical_2013,zanin_permutation_2012} provide a comprehensive overview of developments in the field and related applications.

The features of permutation entropy, mentioned above, make it particularly well-suited for long time series such as those collected from EEG or ECG machines \cite{li_multiscale_2010, liu_multiscale_2017}.
Extensions of permutation entropy such as creating a spectorgram-like visualization of permutation entropy by considering patterns defined by  $x_{t}, x_{t+d}, x_{t+2d}$ for some various $d \geq 1$ are able to even further highlight subtle changes in behavior, even for periodic data.  This method was used to characterize sleep stages from EEG data and matched the expert annotations almost exactly \cite{bandt_permutation_2016}. Similar scale data sets are becoming increasingly available in the current big data paradigm and permutation methods are well positioned to contribute to comprehensive and meaningful analyses.

Another well-motivated application of permutation entropy appears in the context of economic markets.  According to economic theory, an efficient market is one in which price histories cannot predict future behavior, and thus the market is described by a random walk \cite{fama_behavior_1965,fama_random_1995}.  Thus, the proximity of a particular market to the random walk model serves as a proxy for market efficiency.   Observed market inefficiencies can be caused by communication barriers, unfair competition, momentum, and   calendar year effects including the release or announcement of new product lines, among others.  As a result, quantifying inefficiency over time and comparing relative inefficiency between markets is an important, longstanding question in finance \cite{fama_behavior_1965}. 

To distinguish developed and emerging markets, the authors of \cite{zunino_forbidden_2009} use permutation entropy on the changes in stock prices (returns) to measure the independence of these steps.  Other economics researchers used similar methods to evaluate market volatility directly \cite{lim_rapid_2014}.  The approach presented in this paper is motivated in part by these recent applications of permutation entropy and forbidden patterns to financial time series.  In particular, we show that measures of divergence from null models motivated by economic theory can give useful measures of complexity in this setting.  

We accomplish this by extending the interpretation of permutation entropy as a KL divergence from white noise to measuring deviation from a specified null model.  In this paper, we extend permutation-based methods to include a null model for the distribution of patterns in the setting of random walks, which describes several types of observed data, particularly those from economics.  

\subsection{Forbidden Patterns} An early application of permutation patterns arose in iterated function models of dynamical systems. In this setting, the number of distinct patterns in the time series, $X=\{x,f(x),f(f(x)),\ldots\}$, contains information about the complexity of the system itself. Indeed, Bandt, Keller and Pompe showed \cite{bandt_keller_pompe_2002} that any time series defined by iterating a piece-wise monotone map of the interval has forbidden patterns, that is, patterns that never appear in this way.  Moreover, they showed that the logarithm of the exponential growth rate of the number of patterns is equivalent to the topological entropy, an important measure of complexity from the study of manifolds and discrete-time dynamics.  

Following this result, a measure of complexity of a given time series arising from  discrete-time dynamical systems is defined by counting the number of distinct patterns, usually of some fixed length $3 \leq n \leq 7$ \cite{bandt_permutation_2002}. For example, any time series defined by iterating the full logistic map, $f(x) = 4x(1-x)$ contains at most five patterns of length $3$ ($321$ is forbidden), and at most half of the patterns of length four.  On the other hand, a time series defined the significantly more complex system $f(x) = 10 x \pmod 1$ will eventually contain all patterns of length $n \leq 11$ provided that the initial condition does not fall into a periodic orbit in this time.  

  Although in this paper we are concerned with this dynamically motivated viewpoint, there is now a significant literature surrounding the combinatorial aspects of forbidden patterns.   
Combinatorial descriptions of and enumerations for forbidden patterns in specific iterated functions appear in 
 \cite{amigo_forbidden_2008,archer_allowed_2017,archer_cyclic_2014,
charlier_permutations_2017,elizalde_patterns_2015,elizalde_number_2009}. Additionally, patterns in this context have connections with the study of consecutive pattern avoidance, as outlined in this survey \cite{elizalde_survey_2016}. Several related techinques and approaches can be found in the book \cite{amigo_permutation_2010}. As we will see below, the combinatorial approach of strict pattern avoidance is not applicable to time series generated by random walks, since all patterns occur with non-zero probability (cf. Proposition 2).  

The forbidden pattern metric suggests that time series data that exhibit relatively short forbidden patterns contain a deterministic element.  However, simply counting the forbidden patterns does not describe the entire system, as some patterns may only be missing due to data limitations, such as  relatively short length.  In particular, as shown in Table \ref{tab:length2000}, in a time series of 2000 uniformly random numbers 48 of the $6!=720$ total patterns do not appear.  In the limit, we expect all patterns of length $6$ will appear with the same relative frequency (cf. Proposition 1), and so these 48 patterns are not forbidden in the sense of Bandt-Keller and Pompe; such patterns are often referred to as ``missing patterns" in the literature \cite{amigo_true_2007, zanin_permutation_2012}. Conversely, noise in observational data means that we might observe more patterns than should appear based on the model. Thus, strictly counting the number or proportion of forbidden patterns can be a misleading measure.  

As mentioned in the introduction, applications of these methods have found uses in a variety of fields \cite{keller_permutation_2017, riedl_practical_2013,zanin_permutation_2012}. One example that we will revisit throughout this paper is the appearance of patterns in financial time series as discussed in \cite{zunino_forbidden_2009}. In that paper, the number forbidden patterns is used to detect stock market inefficiencies with the understanding that systems with more forbidden patterns are more deterministic, and thus more inefficient. Similarly, in \cite{lim_rapid_2014} a related method is used to study variability in world economics markets surrounding the 2008 financial crisis. 

Although simple to compute, measures based solely on counting forbidden patterns are heavily influenced by noisy data, as non--existence is a very strict criterion for each pattern. Although thresholding or other preprocessing techniques could be used to improve this method for real data, other related measures have grown in popularity. Thus, most recent applications of permutation methods to time series use a related measure, {\em Permutation Entropy}, which is computed from the distribution of patterns that occur, rather than making measurements defined by the strictly binary forbidden/allowed distinction.

\subsection{Permutation Entropy} 
Currently, the most commonly used metric on pattern distributions in time series is the permutation entropy, originally described in \cite{bandt_permutation_2002}.  For a time series $X=\{X_i\}$ and fixed integer $n$ this measure is defined to be the Shannon entropy for the distribution of ordinal patterns of length $n$ that occur in $X$ and is defined to be \cite{bandt_permutation_2002}
$$\PE_n(X) = - \frac{1}{\log(n!)} \sum_{ \pi \in \S_n} p_{\pi} \log(p_{\pi}),$$
where $p_\pi$ represents the proportion of patterns of length $n$ with shape $\pi$ and the logarithm here, and throughout this paper, is taken base 2. 
The following table gives the permutation entropy and number of forbidden patterns in several different types of data sets for small values of $n$. The data is fully described in Section 1.5 and contains both empirical and simulated time series. Of particular interest is the fact that missing patterns appear in all data sets for $n=6$, even those that are guaranteed (cf. Propositions 1 and 2) asymptotically to contain all patterns. Additionally, notice that the permutation entropy values are quite large for many of the noisy and random data sets. 
\renewcommand{\arraystretch}{1.2}
\begin{table}[!h] 

\begin{tabular}{|c||c|c|c||c|c|c|c|}
\hline
 Data & \multicolumn{3}{|c||}  {Forbidden Patterns}& \multicolumn{3}{c|}  {Permutation Entropy}  \\
& \multicolumn{3}{|c||}   {$n = 4$  \hspace{.12 in} $n = 5$  \hspace{.12 in}  $n = 6$ }&  \multicolumn{3}{c|}   {$n = 4$  \hspace{.25 in} $n = 5$  \hspace{.25 in}  $n = 6$ }  \\ 
 \hline
 RAND & 	\rule{.125in}{0pt}	 0 \rule{.125in}{0pt}&\rule{.125in}{0pt}	 0 \rule{.125in}{0pt}& 48&      			0.998919 & 0.991985 & 0.969542 \\    
 \hline                          
 NORM RW &		 0& 0& 190&     			 0.942041 & 0.915789 & 0.875277 \\   
\hline                           
N-DRIFT RW &	 0& 0& 207&     			 0.932001 & 0.900281 & 0.857146 \\   
\hline                           
UNIF RW &		  0& 0& 216&    			 0.929774 & 0.898946& 0.854548\\    
\hline                           
MEX & 			   0& 0& 129&   			0.965306 &  0.952283 & 0.92578 \\     
\hline                           
NYC & 			 0& 0& 115&     		0.961983 & 0.950457& 0.923901 \\          
\hline                           
SP500 &			 0& 0& 199&     		 0.937607 & 0.906991 & 0.862654 \\        
\hline                           
GE &			 0& 2& 210&     		 0.936839 & 0.905735 & 0.863104 \\       
\hline                           
HEART & 		 0& 8& 344&     			0.847181 & 0.813425 & 0.777208 \\     
\hline                           
SIN & 			 14& 106& 702& 		 0.702098 & 0.540424 & 0.422023 \\        
\hline
\end{tabular}
\caption{Computations of permutation entropy and the number of forbidden patterns for a range of time series of length $N = 2000$ described in Section \ref{sec:data}. Notice that all data sets exhibit forbidden patterns even though in the limit several of the rows 1--4 should have all patterns (Propositions 1 and 2). } 
\label{tab:length2000}
\end{table} 

 When a time series is defined by iterating a piece-wise monotone interval map $f$, the permutation entropy of the time series coincides with the Kolmogorov-Sinai entropy of $f$ \cite{keller_relation_2012, keller_kolmogorovsinai_2010}.  Thus, as the number of forbidden patterns is a permutation analog of the topological entropy of $f$, the permutation entropy is an analog of the Kolmogorov-Sinai entropy of $f$ \cite{bandt_keller_pompe_2002}.  However, most time series data that we encounter is not assumed to be derived from an iterated function, even with a noisy model.

For a time series whose values are 
drawn independently from a given distribution, each pattern of length $n$ asymptotically appears with the same relative frequency, see Proposition \ref{prop:randomness}.  Such a time series is considered to be of maximal entropy and has expected permutation entropy equal to 1 as the number of time steps goes to infinity.

This motivates a recently introduced, alternative interpretation of permutation entropy, as the Kullback-Leibler divergence (KL divergence) of the deviation of the empirical distribution from that of white noise (see \cite{bandt_permutation_2016, keller_permutation_2017} for some exposition about this perspective). The KL divergence for the distribution of patterns in $Z$ from those in $Y$ is defined by:

\begin{equation*}\label{eq:KL} \DKL_n(X||Z) =  \frac{1}{\log(n!)} \sum_{ \pi \in \mathcal{S}_n} P_X(\pi) \log\left(\frac{P_X(\pi)}{P_Z(\pi)}\right). \end{equation*}
The relationship between permutation entropy and the Kullback-Leibler divergence 
of the distribution of patterns in the time series from the uniform distribution, $U$, is 
$$\DKL_n(Z || U) = \frac{1}{\log(n!)}\left(\log(n!) + \sum_{\pi \in \S_n} P_{Z}(\pi) \log(P_Z(\pi))\right) = 1 - \PE_n(Z) .$$

The formulation of permutation entropy in terms of the KL divergence from the expected behavior of white noise motivates our approach in this paper since many types of time series, particularly those arising in financial contexts, exhibit characteristic behavior of their distributions of patterns that is highly non--uniform. Our purpose here is to quantitatively explain this difference and provide null models that more closely approximate the distributions seen in actual data.

\subsection{Notation and Terminology}

For consistency,  we describe the notation that we will use throughout this paper. Given an ordered list of values $x_1,x_2,\ldots, x_n$ with $x_i\neq x_j$ for all $i\neq j$  we define the associated permutation $\st(x_1,x_2,\ldots,x_n)=\pi\in S_n$ such that $x_{\pi^{-1}(1)} < x_{\pi^{-1}(2)}< \ldots < x_{\pi^{-1}(n)} $. This is also called the ordinal pattern of $x_1,x_2,\ldots, x_n$. Given a time series $X = \{x_1, x_2, \ldots, x_N\}$, we represent the ordinal pattern of length $n$, beginning at time $t$, by $\st(X,n,t)$.

In this paper we are concerned with the distribution over patterns rather than the specific time of occurrence of any individual pattern since, as described above, the distribution of patterns in a time series $X$ contains important information about the underlying dynamics. For a fixed time series $X$ and permutation $\pi\in S_n$, we denote the empirical proportion of occurrences of the pattern $\pi$ in $X$ by
 $$p_{\pi} := \frac{| \{ i : \st(X_i, X_{i+1}, \ldots, X_{i + n - 1}) = \pi\}|}{N - n + 1}.$$

Similarly, to a sequence of independent random variables, $\{Z_{i}\}_{i = 1}^n$, we define the expected proportion of occurrences of $\pi\in S_n$ by
$$\P_Z(\pi) = \P(\st(Z_1, Z_2, \ldots, Z_n) = \pi) = \P(Z_{\pi^{-1}(1)} < Z_{\pi^{-1}(2)}< \ldots < Z_{\pi^{-1}(n)}),$$
noting that by independence the starting point does not change the probability. 
Thus, for a long time series, $X_t$, whose values are determined by drawing a value at random according to $Z_t$, we expect $p_{\pi} \approx \P_Z(\pi)$. Additionally, by Proposition 1, we note that if the $\{Z_{i}\}_{i = 1}^n$ are independent and identically distributed continuous random variables, then $\P_Z(\pi) = \frac{1}{n!}$ for all $\pi \in \S_n$.  Thus, the distribution of patterns in white noise (i.e. a randomly generated time series) is approximately uniform and converges to the uniform distribution as the length of the time series goes to infinity.   

We primarily focus on the distribution of patterns in random walks $Z= \{Z_i\}_{i = 1}^\infty$ whose steps $\{Y_i\}_{i = 1}^\infty$ are independent and identically distributed continuous random variables with $Z_{i} = \sum_{j = 1}^{i-1} Y_{j}$.  Since the probabilities $\P_Z(\pi)$, for $\pi \in \S_n$ only involve the first $n$ random variables, it will be enough to consider finite random walks, $\{Z_i\}_{i = 1}^n$.  
If there are no requirements on the distributions of the steps $\{Y_i\}$ we say that this is an \textit{arbitrary random walk} while if the steps, $\{Y_i\}$, are symmetric random variables we will say that $Z$ is a \textit{symmetric random walk}.

In this paper, we focus on the properties of two particular random walk null models based on standard step distributions. 
 When the steps $\{Y_i\}$ are normally distributed, we refer to this as a \textit{random walk with normal steps}, with parameters $
 \mu$ and $\sigma$. 
When the steps $\{Y_i\}$ are uniformly distributed on the interval $[b-1, b]$, with $0<b<1$, we refer to this as a \textit{random walk with uniform steps}.   The parameter specifying the distribution is $\P(Y_i > 0) = b$.   Due to the scale invariance of the permutation measure, it suffices to consider an interval of unit length. Since each of the $Y_i$, are identically distributed, we will sometimes drop the subscript when referring to their distributions.

\subsection{Contributions}   Our purpose in this paper is to describe the distributions of ordinal patterns of random walk null models for time series data in order to derive a  corresponding KL measure generalizing permutation entropy. These models are motivated by the KL divergence definition of permutation entropy described in Section 2.2 and domain specific hypotheses about the random behavior of time series data. In the next section we describe the theoretical properties of these models, including the expected distributions, which allow us to define a KL divergence to the derived values. Next, we describe a metric, based on recent work of Martinez and Elizalde \cite{martinez_frequency_2015}, that measures how well a given distribution matches any random walk model. We conclude by applying the new methods to a wide variety of data sets to demonstrate their advantages and applicability. 

\subsection{The Data}\label{sec:data}
Throughout this paper we use several example data sets to evaluate our methods and compare to traditional approaches. Unless otherwise specified, these time series have $N=2000$ data points.  This data includes synthetic random values as well as empirical data from economics, ecology, and medicine. Below we describe the key features of the data and the abbreviations that we use throughout the paper. Plots of the time series are displayed in Appendix A.
\begin{itemize}
\item (RAND):  A sequence of 2000 uniform random numbers drawn from $[0,1]$. 
\item (NORM RW): A simulated random walk whose steps are drawn at random from the standard normal distribution, $(\mu, \sigma) = (0, 1)$.  
\item (N-DRIFT RW): A simulated random walk whose steps are drawn at random from the normal distribution with $(\mu, \sigma) = (0.701832, 14.945)$; this is the normal curve fitted to the returns in the S\&P500 data below.  
\item (UNIF RW): A simulated random walk whose steps are drawn uniformly at random from the uniform distribution on the interval $[-.5, .5]$.  
\item (U-DRIFT RW): A simulated random walk whose steps are drawn uniformly at random from the uniform distribution on the interval $[-.35, .65]$. 
\item  (SP500): The daily closing values of the S\&P500 from January 24, 2009 to December 31, 2016. Data provided by Morningstar  and accessed through \cite{Mathematica}.
\item (MEX): Average daily temperatures in Mexico City from June 20, 2011 to December 31, 2016. Data provided by the World Meterological Organization   through \cite{Mathematica}.  
\item (NYC): Average daily temperatures in New York City from June 20, 2011 to December 31, 2016.    From the National Oceanic and Atmospheric Administration through \cite{Mathematica}.
\item (HEART): Instantaneous heart rate measurements taken at $.5$ second intervals collected at  MIT \cite{goldberger_nonlinear_1991}. 
\end{itemize}

In all cases, random values are generated using Mathematica's \cite{Mathematica} pseudo-random number generator and all historical market closing values are provided by Morningstar through Mathematica. In the final section, we use the daily closing prices of the S\&P500, Apple (AAPL), Amazon (AMZN), Bank of America (BAC), General Electric (GE), Coca Cola (KO), and United Parcel Service of America (UPS) for  trading days from January 1, 2002 to January 1, 2017 ($N  = 3777$).  Finally, for a longitudinal test, we use daily closing prices of the S\&P500 from January 1, 1958 until January 1, 2017 ($N=14348$).

\section{Distributions of Patterns in Random Walks} 

In this section, we establish some of the important properties of the distribution of patterns for random walk null models.  For the uniform and symmetric normal random walk models, we give the distribution of patterns of length $n = 3, 4$ in Table \ref{tab:normunifprobabilites} and show how these values can be computed for larger $n$ in Proposition 5. 

\subsection{Comparison to Forbidden Patterns and Permutation Entropy}  We begin by showing that any data model whose values are i.i.d. random variables gives rise to the uniform distribution over permutations. 
\begin{prop}\label{prop:randomness} If $Y$ is a sequence of i.i.d. continuous random variables, then for any $\pi \in \mathcal{S}_n$, we have 
$$\P_Y(\pi) = \frac{1}{n!}.$$
\end{prop}
\begin{proof} Let $\pi \in \S_n$; we will show that $\P_Y(\pi) = \P_Y(12 \ldots (n{-}1)n)$.  Since $Y_i$ and $Y_j$ are i.i.d., we have $\P(Y_i < Y_j) = \P(Y_j < Y_i)$, and so transposing any pair of variables does not change the probability of the event.  It follows that 
$$\P_Y(\pi) = \P(Y_{\pi^{-1}(1)} < Y_{\pi^{-1}(2)} < \ldots < Y_{\pi^{-1}(n)}) = 
 \P(Y_1 < Y_2 < \ldots < Y_n )  = \P_Y(12 \ldots (n{-}1)n).$$
Therefore, all permutations of a fixed length occur with the same probability, and so $P_Y(\pi) = \frac{1}{n!}$.  
\end{proof}

In particular, Proposition \ref{prop:randomness} implies that for any  time series $Y = \{Y_i\}_{i = 1}^\infty$ generated by a random process, for each $\pi \in \S_n$, we expect the relative frequency of $\pi$ to approach $\frac{1}{n!}$  as $N \rightarrow \infty$.  
The next results describe how the behavior of random walk models differ from the forbidden pattern and permutation entropy measures by showing that in such a model there are no forbidden patterns (Proposition 2) and that the distribution of patterns is never uniform (Propositions 3 and 4).

\begin{prop} \label{prop:noforbidden} If $Z$ is a normal or uniform random walk such that $\P(Y > 0) \notin \{0, 1\}$ then for any $\pi \in \mathcal{S}_n$, we have
$$\P_Z(\pi) > 0.$$
\end{prop}

\begin{proof} Let $b=\P(Y > 0)$, consider the case when $b \leq \frac{1}{2}$, the other case follows a similar argument.    
 Define a collection of intervals by $I_j= \left(\frac{(j-1)b}{n}, \frac{jb}{n}\right)$ with midpoint $m_j=\dfrac{(2j-1)b}{2n}$, for $1 \leq j \leq n$.  Notice that $ \{I_j\}_{j=1}^n\subset\left(0,b \right)$. We claim that for all positive integers $i$ and pairs $1\leq j, k\leq n$ we have  $$\P(Z_{i+1}\in I_k | Z_{i}\in I_j)>0.$$ 

Suppose that $j < k$, the other case follows a parallel argument.  Let $Z_{i+1} \in I_j$, and so either a) $Z_{i} \leq m_j$ or b) $Z_{i}> m_j$.  In case a), we have $\P(Z_{i+1}\in I_k |Z_{i}\in I_j)>\P( \frac{(k-j)b}{n} <Y_i< \frac{kb}{n}-m_j )$ and in case b) we have $\P(Z_{i+1}\in I_k | Z_{i}\in I_j)>\P( \frac{kb}{n}-m_j <Y_i< \frac{(k-j+1)b}{n} )$.

In the case that $Y_i$ is normal, it is clear that $\P( \frac{(k-j)b}{n} <Y_i< \frac{kb}{n}-m_j )$ and $\P( \frac{kb}{n}-m_j <Y_i< \frac{(k-j+1)b}{n} )$ have positive probability since every interval has positive probability.  Hence, $\P(Z_{i+1}\in I_k | Z_{i}\in I_j)>0.$

In the case that $Y_i$ is uniform on $[1-b, b]$, we have  $\P( \frac{(k-j) b}{n}<Y_i< \frac{bk}{n}-m_j) > 0$ because $\frac{(k-j) b}{n} \in [1-b, b]$.  Similarly $\P( \frac{kb}{n}-m_j <Y_i< \frac{(k-j+1)b}{n}) >0$ since $\frac{(k-j+1)b}{n} \in [1-b, b]$.   It follows that $\P(Z_{i+1}\in I_k | Z_{i}\in I_j)>0.$

We now write
$$\P_Z(\pi) = \P(\st(Z_1, Z_2, \ldots, Z_{n}) = \pi) \geq \P(Z_{i} \in I_{\pi(i)} \text{ for all } 1  \leq i \leq n).$$
Notice that the events $\P(Z_{i} \in I_k | Z_{i-1} \in I_j)$ and $\P(Z_{i'} \in I_k | Z_{i'-1} \in I_j)$ are independent since they only depend on $Y_i$, $Y_{i'}$, respectively, which are themselves independent for $i \neq i'$.  We now have
$$\P(Z_{i} \in I_{\pi(i)} \text{ for all } 1  \leq i \leq n) = \P(Z_1 \in I_{\pi(i)}) \prod_{i = 2}^n \P(Z_{i} \in I_{\pi(i)} | Z_{i-1} \in I_{\pi(i)-1}) > 0,$$  
and conclude that $\P_Z(\pi)> 0$.  
\end{proof}

In particular, Proposition \ref{prop:noforbidden} implies that in a time series $Z = \{Z_i\}_{i = 1}^\infty$ defined by a normal or uniform random walk model, we expect that
$$\lim_{N \rightarrow \infty} \# \text{Forbidden Patterns}(\{Z_i\}_{i=1}^N) \rightarrow 0.$$ 

Thus, as mentioned previously, for the random walks in Table \ref{tab:length2000}, the patterns that do not appear in the time series $X$ are not forbidden, but merely ``missing."  Note that this is another example of the divergence 
 of random walk models from the traditional methods that were motivated by one-dimensional dynamical systems.

\begin{prop} \label{prop:highestfreq} In an arbitrary random walk, the ordinal pattern occurring with the highest frequency is $12\dots (n{-}1)n$ if $\P(Y>0) \geq \frac{1}{2}$ and $n(n{-}1) \ldots 21$ if $\P(Y>0) \leq \frac{1}{2}$.   
\end{prop}\begin{proof}    Let $\pi \in S_n$ and let $\{i_1, i_2, \ldots, i_k\}$ be the descent set of $\pi$.  It follows that
$$\mathbb{P}(\pi) \leq \mathbb{P}(\pi \in S_n: \text{des}(\pi) = \{i_1, i_2, \ldots, i_k\}) = \P(Y<0)^k \P(Y>0)^{n-k-1}.$$   
Moreover, when $\P(Y>0) \geq \frac{1}{2}$, equality only occurs for $\pi = 12 \ldots n$; and when $\P(Y<0) \leq \frac{1}{2}$ equality only occurs for $\pi = n (n{-}1) \ldots 21$.  \end{proof}

This result mirrors recent work on permutons \cite{kenyon_permutations_2015}, where a similar result is obtained in the limiting case. For our purposes, this is enough to show that the distribution of patterns in a random walk can never match the uniform distribution derived from i.i.d. random data.

\begin{prop} For any random walk $Z$, the distribution of patterns of length $n$ (for $n \geq 3$) is not the uniform distribution. 
\end{prop}

\begin{proof}  We have that \begin{eqnarray*}
\mathbb{P}_Z(12\ldots (n{-}1)n) &=&   \prod_{i = 1}^{n-1} \mathbb{P}(Y_i > 0) = \P(Y>0)^{n-1} \\  \mathbb{P}(n (n{-}1) \ldots 2 1) &=& \prod_{i = 1}^{n-1} \mathbb{P}(Y_i < 0) = \P(Y <0)^{n-1}.  \end{eqnarray*}
In particular $\mathbb{P}_Z(12 \ldots (n{-}1)n) = \mathbb{P}_Z(n(n{-}1) \ldots 21)$ implies that $\P(Y>0) = \P(Y<0)=\frac{1}{2}$.  In such a case, we would obtain 
$$\mathbb{P}_Z(12 \ldots (n{-}1)n) = \mathbb{P}_Z(n (n{-}1) \ldots 21) = \frac{1}{2^{n-1}}.$$
If the patterns of length $n$ had the uniform distribution, then $\mathbb{P}_Z(\pi) = \frac{1}{n!}$ for each $\pi \in \mathcal{S}_n$.  But this is impossible because $\frac{1}{2^{n-1}} \neq \frac{1}{n!}$ when $n \geq 3$.  
 \end{proof}  

 This result implies that for data sets derived from a random walk, the distribution of ordinal patterns must differ from the uniform distribution enforced by random data. Figure 1 shows examples of the characteristic shapes of distributions that arise from random walks. The symmetry apparent in the first two distributions of patterns is not coincidental. It arises from the symmetry in the definition of the random walks which we discuss more fully in Section 3.


\begin{figure}[!h]

\begin{tikzpicture}[scale = .65]
\begin{axis} [ybar, bar width = 4.5pt,
    ymin=0,
      ymax=.15,
         ylabel = $P_Z(\pi)$,
                  xlabel=$\pi$,
     ytick={.05,.1, .15},
        xticklabels={, , , , , },
     yticklabels={$.05$, $.10$, $.15$},
    xticklabel style={anchor=east, rotate= 90},
 width = 8 cm, 
 height = 6 cm, 
]
\addplot[ybar, color = orange,  fill=orange]
table[x = x, y = unif] {\histfourmodels};
\end{axis} 
\end{tikzpicture}
\quad
\begin{tikzpicture}[scale = .65]
\begin{axis} [ybar, bar width = 4.5pt,
    ymin=0,
      ymax=.15,
         ylabel = $P_Z(\pi)$,
                  xlabel=$\pi$,
     ytick={.05,.1, .15},
     yticklabels={$.05$, $.10$, $.15$},
        xticklabels={, , , , , },
    xticklabel style={anchor=east, rotate= 90},
 width = 8 cm, 
 height = 6 cm, 
]
\addplot[ybar, color = orange,fill=orange]
table[x = x, y = stn] {\histfourmodels};
\end{axis} 
\end{tikzpicture}
\quad
\begin{tikzpicture}[scale = .65]
\begin{axis} [ybar, bar width = 4.5pt,
    ymin=0,
    xticklabel style={anchor=east, rotate= 90},
    ymax=.30,
     ytick={.05,.1, .15, .20, .25, .30 },
         ylabel = $P_Z(\pi)$,
         xlabel=$\pi$,
        xticklabels={, , , , , },
     yticklabels={, $.10$,  , $.20$, , $.30$},
 width = 8 cm, 
 height = 6 cm, 
]
\addplot[ybar,color = orange, fill=orange]
table[x = x, y = drift] {\histfourmodels};
\end{axis} 
\end{tikzpicture}

\caption{The distribution of patterns of length $n = 4$, listed in lexicographical order, for (a) the normal random walk with $\mu = 0$, (b) the uniform random walk with $\P(Y > 0) = .5$ and (c) the uniform random walk with $\P(Y > 0) = .65$, computed using Proposition 5.} 
\label{fig:differentsteps}
\end{figure}
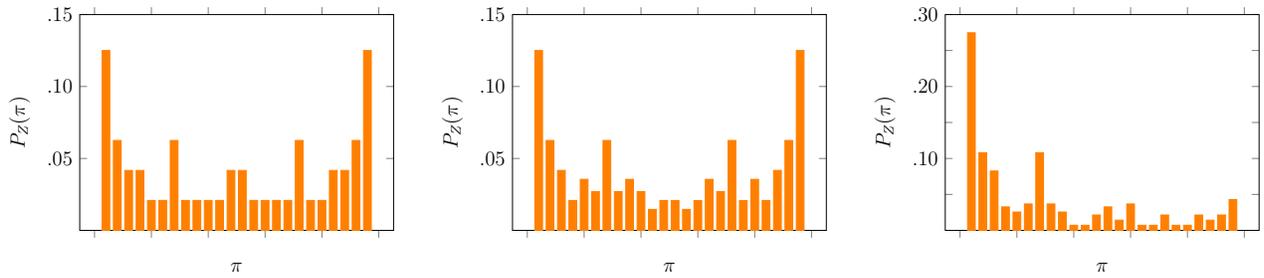

\subsection{Random Walk Distributions}
We next describe how to compute the expected distributions of ordinal patterns for uniform and normal walks. These values allow us to measure the KL divergence from an empirical data set to the expected values under a random walk null model. 

\begin{prop}\label{prop:regionintegrate}  For a uniform or normal random walk $Z$, the value $\P_Z(\pi)$ for $\pi \in \S_n$ can be interpreted as a volume of a region in an $(n{-}1)$-dimensional surface and bounded by certain hyper-planes through the origin (see Figure \ref{fig:regionsintegrate}).
\end{prop}

\begin{proof} Let  $b = \P(Y>0)$.   We graphically represent the joint distribution $\{Y_1, Y_2, \ldots, Y_{n-1}\}$ as an $(n{-}1)$-dimensional region.  Thus, the probability density function is uniform on the $(n{-}1)$-dimensional cube $[b-1, b]^{n-1}$, which is partitioned by patterns according to the relative order of $(Z_1, Z_2, \ldots, Z_{n}) = (0, Y_1,  \ldots, Y_1 + Y_2 + \ldots + Y_{n-1})$.   

 For example, consider $\P(\pi) =\P(Z_{\pi^{-1}(1)} < Z_{\pi^{-1}(2)}< \ldots < Z_{\pi^{-1}(n)})$.  If $\pi^{-1}(i + 1) > \pi^{-1}(i)$, the inequality $Z_{\pi^{-1}(i +1)} <  Z_{\pi^{-1}(i)}$  becomes $$0 < Y_{\pi^{-1}(i + 1)} + Y_{\pi^{-1}(i + 1) - 1} + \ldots Y_{\pi^{-1}(i) + 1}.$$ 
 Similarly, if $\pi^{-1}(i+1) < \pi^{-1}(i)$,  the inequality $Z_{\pi^{-1}(i +1)} <  Z_{\pi^{-1}(i)}$  becomes $$0 < Y_{\pi^{-1}(i)} + Y_{\pi^{-1}(i)  - 1} + \ldots + Y_{\pi^{-1}(i+1) + 1}.$$ 
It follows that the regions in the hypercube such that $Z_{\pi^{-1}(1)} < Z_{\pi^{-1}(2)}< \ldots < Z_{\pi^{-1}(n)}$ are bounded by certain hyperplanes of the form $a_1 x_1 + a_2 x_{2} + \ldots + a_{n-1} x_{n-1} = 0$, for $a_1, a_2, \ldots, a_{n-1} \in \{0, 1\}$, and so, in all cases $\P_Z(\pi)$ can be interpreted as a volume of a region in an $(n{-}1)$-dimensional cube bounded by hyperplanes through the origin. \end{proof}

\begin{figure}[!h]

\begin{tikzpicture}[scale = .85]

\draw[-, line width=2pt] (4*.65, 4*-.35)--(4*.65, 4*.65);
\draw[-, line width=2pt] (4*-.35, 4*.65)--(4*.65, 4*.65);
\draw[-, line width=2pt] (4*-.35, 4*-.35)--(4*.65, 4*-.35);
\draw[-, line width=2pt] (4*-.35, 4*-.35)--(4*-.35, 4*.65);
\draw[-, line width=1pt] (4*-.35, 0)--(4*.65, 0);
\draw[-, line width=1pt] (0, 4*-.35)--(0, 4*.65);
\draw[-, line width=1pt] (4*-.35, 4*.35)--(4*.35,4* -.35);
\draw (1.55, 1.55) node {\scriptsize{123}};
\draw (-.8, -.8) node {\scriptsize{321}};
\draw (-.8, .35) node {\scriptsize{312}};
\draw (-.8, 1.55) node {\scriptsize{213}};
\draw (.35, -.8) node {\scriptsize{231}};
\draw (1.55, -.8) node {\scriptsize{132}};
\draw (1.85, -2.1) node {$Y_1$};
\draw (2.55, -1.75) node {$b$};
\draw (-2.25, 1.75) node {$Y_2$};
\draw (0, -1.75) node {$0$};
\draw (-1.95, -1.25) node {$b{ -}1$};
\draw (-.95, -1.75) node {$b{ -}1$};
\draw (-1.85, 2.5) node {$b$};
\draw (-1.85, 0) node {$0$};
\end{tikzpicture}
\qquad
\includegraphics[height=1.66in]{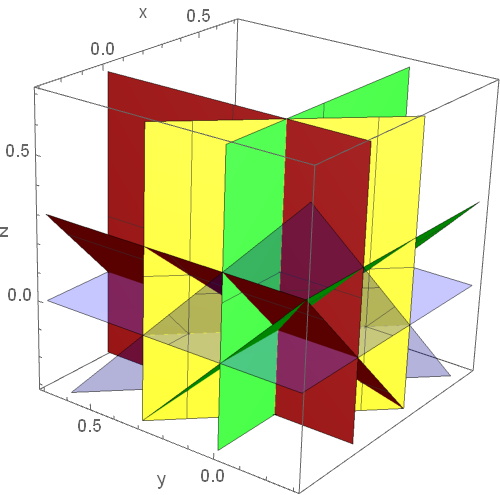}

\caption{The regions of integration for patterns in uniform random walks for (a) $n = 3$ and (b) $n = 4$, sketched here for $ b= .65$.} 
\label{fig:regionsintegrate}
\end{figure}

For the random walk with uniform steps and $\P(Y>0) \geq \frac{1}{2}$, the probability of each of the ordinal patterns of length 3 and those of length 4 occurring are given in the Table \ref{tab:normunifprobabilites} in Appendix A. 
For a random walk with normal steps and $\mu = 0$, the probability of each of the ordinal patterns of length 3 and length 4 are computed using Proposition \ref{prop:regionintegrate} and the spherical symmetry of the multi-variate normal distribution and the area of spherical triangles.  In particular, when $\mu = 0$, the spherical symmetry of sums of normally distributed random variables tells us that the distribution of patterns is independent of the variance, but this is not the case when $\mu \neq 0$.  

This result allows us to determine the expected behavior of the distribution of ordinal patterns under the assumption that the data was generated by a particular random walk null model. Thus, we can compute the KL divergence between the expected value and  empirical data to measure the portion of the behavior explained by the random walk model.  


%
%

\subsection{Examples}

We conclude this section with two examples highlighting the differences between our models and the i.i.d. model that underlies permutation entropy. This allows us to demonstrate that for some data sets, the distributions derived from a random walk model matches empirical data quite closely compared to the uniform distribution.

 We begin by constructing a time series of length 2000 from a uniform random walk (U-DRIFT RW) by fixing $b=\P(Y>0)=.65$ and comparing the distribution of patterns of length 4 to the values derived from Proposition 5 as well as the uniform values of $\frac1{24}$. Figure 3 displays these results, the observed distributions are plotted in blue (on both graphs) while the gray bars represent the expected random walk distribution (left) and uniform distribution (right).
 
  As expected, the observed values match the null model distributions much more closely than the uniform distribution. Note that the expected and observed values on the left do not match exactly because the emprical time series has finite length. This is a common feature of time series data that is observed throughout this paper. 
  
\begin{figure}[!h]
\begin{tikzpicture}[scale = .8]
\begin{axis} [ybar, bar width = 8pt, bar shift=0pt,
    ymin=0,
  ytick={.05, .10, .15, .20, .25},
     yticklabels={$.05$, $.10$, $.15$, $.20$, $.25$},     
    xticklabel style=,
       xticklabels={, , , , , },
    ylabel = $p_{\pi}$,
 width = 11 cm, 
 height = 6 cm, 
]
\addplot[ybar, fill opacity=0.5, fill=black]
table[x = x, y = trueunif] {\matchinghistunifdrift};
\addplot[ybar, fill opacity=0.5, fill = blue]
table[x = x, y = simulateddrift] {\matchinghistunifdrift};
\end{axis} 

\end{tikzpicture}
\begin{tikzpicture}[scale = .8]
\begin{axis} [ybar, bar width = 8pt, bar shift=0pt,
    ymin=0,
  ytick={.05, .10, .15, .20, .25},
     yticklabels={$.05$, $.10$, $.15$, $.20$, $.25$},  
        xticklabels={, , , , , },   
    xticklabel style=,
    ylabel = $p_{\pi}$,
 width = 11 cm, 
 height = 6 cm, 
]
\addplot[ybar, fill opacity=0.5, fill=black]
table[x = x, y = flat] {\matchinghistunifdrift};
\addplot[ybar, fill opacity=0.5, fill = blue]
table[x = x, y = simulateddrift] {\matchinghistunifdrift};
\end{axis} 

\end{tikzpicture}
\caption{The distribution of patterns of length $n = 4$ in U-DRIFT RW, listed in lexicographical order compared to (a) the true distribution of patterns in the uniform random walk on with $\P(Y > 0) = .65$ (see Table \ref{tab:normunifprobabilites}) and (b) the distribution of patterns in white noise.}
\end{figure}
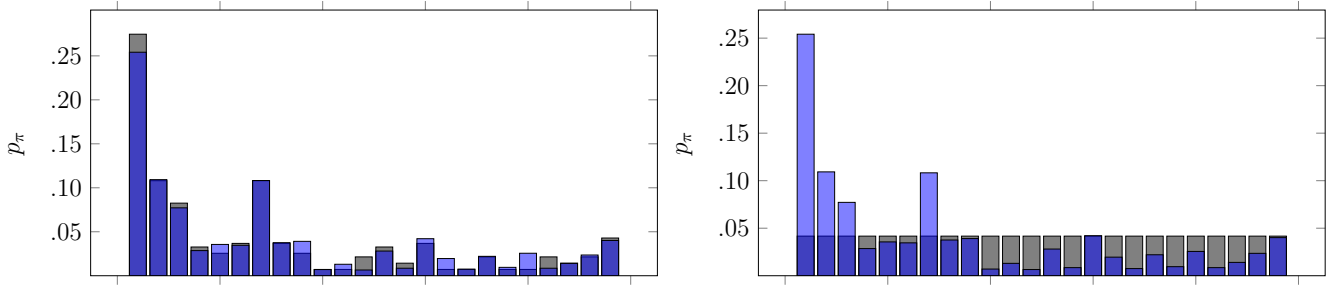

We next consider a similar analysis for economic market data, using the closing prices of the S\& P 500 over a seven year period (SP500). For this example, we need to estimate an underlying distribution. To do this we calculate the sequence of steps $\{X_{t+1} - X_{t}\}_{t = 1}^\infty$ (called the stock \textit{returns})and find the best fit normal curve; in this case obtaining parameters $(\mu, \sigma) = (0.702, 14.945)$. The null model for SP500 is the distribution of patterns for the normal random walk with these parameters. Using a simulated normal random walk we approximate the distribution of patterns for a fixed $n$.  

This null model is shown in Figure 4 for $n=4$ and $n=5$. Note that this data displays a very similar shape to those in Figures 1 and 3 and is highly non--uniform. This reinforces our conclusion that modeling some time series with random walk null models more effectively describes the behavior in this case than permutation entropy.

\begin{figure}[!h]
\begin{tikzpicture}[scale = .8]
\begin{axis} [ybar, bar width = 8pt,
    ymin=0,
  ytick={.05, .10, .15},
     yticklabels={$.05$, $.1$, $.15$},     
   xticklabels={, , , , , },
    ylabel = $p_{\pi}$,
 width = 11 cm, 
 height = 6 cm, 
]
\addplot[ybar, fill=blue]
table[x = x, y = y] {\HistSPNullModelFour};
\end{axis} 

\end{tikzpicture}
\begin{tikzpicture}[scale = .8]
\begin{axis} [ybar, bar width = 1pt,scaled y ticks=false,   
   	 ymin=0,
   	 ymax = .08,
     ytick={.02, .04, .06},
         ylabel = $p_{\pi}$,
     yticklabels={$.02$, $.04$, $.06$},   
        xticklabels={, , , , , },
 width = 10 cm, 
 height = 6 cm, 
]
\addplot[ybar, color = orange, fill = orange]
table[x = x, y = y] {\HistSPNullModelFive};
\end{axis} 
\end{tikzpicture}

\caption{The distribution of patterns, listed in lexicographical order, for the uniform random walk null model for SP500 of length (left) $n = 4$ and (right) $n  = 5$. Note that the distributions are  far from uniform as is characteristic of random walk data. } 
\end{figure}
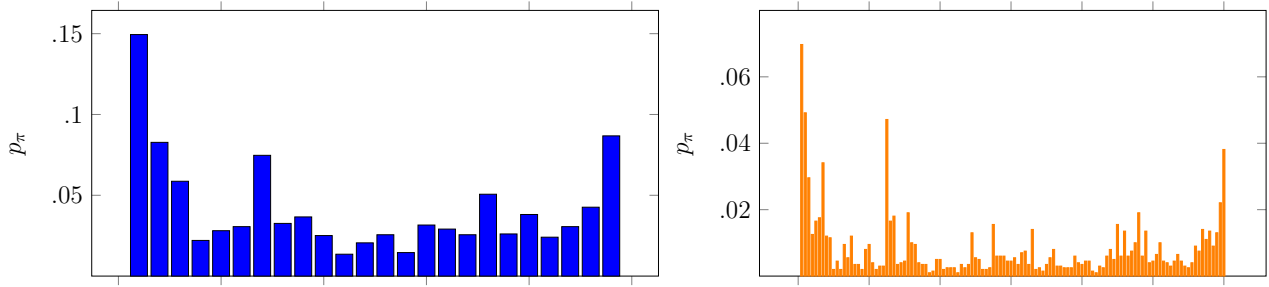

This example was computed with respect to a particular null model, however, there are many  options for selecting the distribution of steps $Y$. A discussion of the possible inferential processes for selecting $Y$ given a particular data set is beyond the scope of this paper. However, for the purposes of comparing to permutation entropy we consider several difference choices of $Y$ and compare their performance to the uniform distribution. These results are summarized in Figure 5 below. 

We compare the distributions derived from the actual SP500 data to three random walk null models: (a)the normally distributed model described above with $(\mu, \sigma) = (0.702, 14.945)$, (b) a uniform model with $\P(Y>0)=p_{12}=.5441$,  and (c) a uniform model fitting the stock returns with $\P(Y>0)=.5279$. The error between the expected values and the empirical values are shown for each permutation in Figure 5. Notice that each of the random walk models significantly outperforms the permutation entropy model on almost all permutations. The sum of squared errors for randomness is 0.0213 and for each of the models is (a) 0.0018, (b) 0.0027, (c) 0.0031.    Although there is some variance among the random walk models, they each convincingly outperform the uniform distribution.



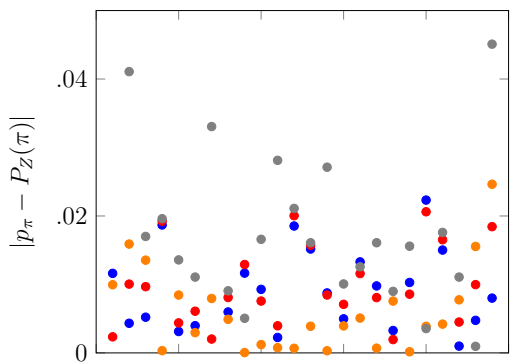
\begin{figure}[!h]
\begin{tikzpicture}[scale = .8]
\begin{axis}[
scaled y ticks=false,
ymax   = .05,
ymin   = 0,
     ytick={-0.04, -0.02, 0, .02, .04},
     yticklabels={$-.04$, $-.02$, $0$, $.02$, $.04$},   
xmax   = 25,
xmin   = 0,
ylabel = $|p_\pi-P_Z(\pi)|$,
xticklabels={, , , , , , }
]

\addplot[color = blue,
fill  = blue,
mark  = *,
only marks]
table[x = pat, y = unifb] {\fittingerrors};

\addplot[color = orange,
fill  = orange,
mark  = *,
only marks]
table[x = pat, y = normalmean] {\fittingerrors};

\addplot[color = red,
fill  = red,
mark  = *,
only marks]
table[x = pat, y = unifmean] {\fittingerrors};

\addplot[color = gray,
fill  = gray,
mark  = *,
only marks]
table[x = pat, y = random] {\fittingerrors};

\end{axis}
\end{tikzpicture}

\caption{Comparison of null model distributions for the SP500 data to the uniform distribution. The difference $|p_\pi-P_Z(\pi)|$ is plotted for each of the four null models: $Y_1=N(0.702, 14.945)$ (orange), $Y_2=U(.5441)$ (blue), $Y_3=U(.5279)$ (red), and the uniform distribution (gray). Values for the null models are derived via Proposition 5. }
\end{figure}

\section{Equality in Any Random Walk}\label{sec:MM}

Although the distributions of permutations under random walk null models are not uniform they are still constrained in some ways by the structure of the models, particularly the assumption of i.i.d. steps. This is reflected in the characteristic distribution shapes displayed in the figures above.  The possible behaviors of these models was recently considered in \cite{martinez_frequency_2015}, giving a classification of permutations that must occur with the same  probability in any random walk model. Here we use related results to characterize distribution in terms of their proximity to the random walk constraints.

The existence of nontrivial equivalence classes of permutations that appear with the same frequency in any random walk is an important distinguishing characteristic of patterns in this context.  
To illustrate the symmetries underlying this feature of permutations in a random walk, we present two results describing constraints that must occur in this setting. 




\begin{prop}\label{prop:sym}  If $Z$ is a random walk with symmetric steps, then 
$$\mathbb{P}_Z(\pi) = \mathbb{P}_Z(\pi^c),$$
where $\pi^c(i):= (n+1) - \pi(i)$ is the complement of $\pi$.  
\end{prop}

\begin{proof} 
Since $Z_j = \sum_{i = 1}^{j-1} Y_i$ is a sum of symmetric random variables, $Z_j$ is itself symmetric. Thus, for any $1 \leq j, k \leq n$, we have $\mathbb{P}(Z_j > Z_k) = \mathbb{P}(Z_j < Z_k)$.  It follows that 
$$\mathbb{P}_Z(\pi) = \mathbb{P}(Z_{\pi^{-1}(1)} < Z_{\pi^{-1}(2)} < \ldots < Z_{\pi^{-1}(n)}) = \mathbb{P}(Z_{\pi^{-1}(n)} < Z_{\pi^{-1}(n{-}1)} < \ldots < Z_{\pi^{-1}(1)}) = \mathbb{P}_Z(\pi^c).$$
\end{proof}

The symmetry condition in Proposition 6 is quite strong.  In particular, it will not apply to real world data containing drift or expected long term gain.  In contrast, Proposition 7 holds for any random walk, regardless of the underlying distribution of steps.

\begin{prop}[Martinez \cite{MeganThesis}] \label{prop:Megan} In an arbitrary random walk, $Z$, we have $$\P_Z(\pi) = \P_Z(\pi^{rc}),$$
where $\pi^{rc}(i) := (n+1) - \pi(n-i)$ is the reverse-complement of $\pi$.  
\end{prop}

\begin{proof}  For the sequence of i.i.d. random variables $\{Y_i\}_{i = 1}^n$, define $Z'_{j} = \sum_{i = 1}^{j-1} Y'_{n+1-i}$.  It follows that $\P_Z(\pi) = \P_{Z'}(\pi)$ for all $\pi \in \S_n$.  For $1 \leq j, k \leq n$, we have $Z_j + Z_{n+1-j}' = \sum_{i = 1}^{n-1} Y_i$ and so  $Z_j < Z_k$ if and only if $Z'_{n+1-j} > Z_{n + 1-k}'$.  It follows that
$$\mathbb{P}(\pi)  = \mathbb{P}(Z_{\pi^{-1}(1)} < Z_{\pi^{-1}(2)} < \ldots < Z_{\pi^{-1}(n)}) $$
$$= \mathbb{P}(Z'_{n+1-\pi^{-1}(1)} > Z'_{n+1 - \pi^{-1}(2)} > \ldots > Z'_{n+1 - \pi^{-1}(n)}) = \P_{Z'}(\pi^{rc}) = \P_Z(\pi^{rc}).$$
 \end{proof}

In particular, Proposition \ref{prop:Megan} explains why the probability of certain permutations, such as $1243$ and $2134$, are equal in each of the distributions considered in Table \ref{tab:normunifprobabilites}.   In \cite{martinez_frequency_2015}, Proposition \ref{prop:Megan} is extended to give a complete characterization of the classes of patterns that appear with the same frequency, regardless of distribution.

In particular, the patterns that are listed in the same line in Table \ref{tab:normunifprobabilites} occur with the same probability in any random walk, regardless of the distribution associated to the steps. The full decomposition into equivalence classes is presented in Table 2 in Appendix B. This explicit decomposition had not been previously computed for $n=4, 5$.

Next we use this structure to define a simple test for determining whether a random walk may be an appropriate choice of model based on these equivalence classes.  For each equivalence class $\Lambda_i \subset \mathcal{S}_n$, of permutations occurring with the same probability in any random walk, define $\mu_i = \frac{1}{\Lambda_i} \sum_{\pi \in \Lambda_i} p_{\pi}$.  We let $g_n(T)$ be total variation from the mean across each equivalence class
$$g_n(T) = \sum_{\Lambda_i \subset S_n} \sum_{\pi \in \Lambda_i}  |p_{\pi} - \mu_i|.$$

Notice that in a random walk, $X = \{X_t\}_{t=1}^N$, and also in white noise, $X' = \{X'_t\}_{t = 1}^N$, we have
$$ \lim_{N \rightarrow \infty} g_n(X) = 0 \text{ and } \lim_{N \rightarrow \infty} g_n(X') = 0.$$
 
Thus, $g_n(X)$ is a measure of the amount that the distribution of permutations that remains unexplained by any random walk model. Figure 6 demonstrates how the value of $g_n(X)$ evolves for a normal random walk and a sequence of i.i.d. randomly generated data points (RAND)
for $n=4$ and $n=5$. As predicted above as $N\rightarrow\infty$ the values of $g_n(X)$ go to zero but that it requires a large number of data points, echoing our comment in Section 2.3 about time series data. This is further supported by the fact that we observe that the random walk and the i.i.d. model appear to converge at the same rate, suggesting that the discrepancy is caused by the finite number of time steps. 

As an example of a model that does not respect these classes, consider a sequence of random variables $Z = \{Z_i\}_{i = 1}^\infty$ and $Z_i = \sum_{j = 1}^{n-1} Y_j$ where the steps $Y_j$ are drawn from:
$$Y_j=\begin{cases} U([-.5,.5]) &Y_{j-1}\geq0 \\
U((0,-Y_{j-1}) &Y_{j-1}<0.\end{cases}$$
In this case the pattern 1243 occurs with a relatively high frequency but the pattern 2134 is forbidden, leading to a large value for $g_n(T)$, in expectation.  Notice that $Z$ is not a random walk because the steps $Y_j$ are not i.i.d. and hence this sequence does not contradict any of our previous propositions.

A limitation of this method is that equivalence classes for $12\ldots(n{-}1)n$ and $n(n{-}1)\ldots21$ consist of a single permutation and so will never contribute to the value of $g_n$. Additionally, as can be seen in Figure 6, the convergence to zero can be slow even for data drawn directly from a null model. An alternative measure is suggested by Proposition \ref{prop:highestfreq}, which gives $\P(12 \ldots (n{-}1)n) = b^{n-1}$ and $\P(n(n{-}1)\ldots 21) = (1-b)^{n-1}$.  Therefore, if a time-series $X$ is modeled by a random walk, we expect 
$$\varepsilon_{12\ldots(n{-}1)n} := p_{12\ldots(n{-}1)n} - (p_{12})^{n-1} \approx 0 \text{ and } \varepsilon_{n(n{-}1)\ldots 21}  := p_{n(n{-}1)\ldots21} - (p_{21})^{n-1} \approx 0.$$

In Section \ref{sec:KLstocks}, we calculate $\varepsilon_{12\ldots(n{-}1)n}$ and $\varepsilon_{n(n{-}1)\ldots21}$ for several stocks and discover the effects of market {\em momentum} in the data.

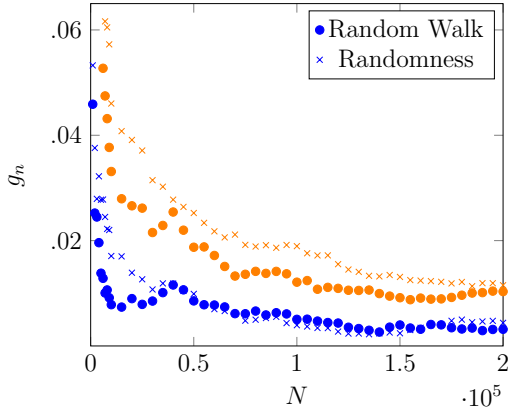
\begin{figure}

\begin{tikzpicture}[scale = .8]
\begin{axis}[
scaled y ticks=false,
ymax   = .065,
ymin   = 0,
     ytick={.02, .04, .06},
     yticklabels={$.02$, $.04$, $.06$},   
xlabel = $N$,
xmax   = 200000,
xmin   = 0,
ylabel = $g_n$
]

\addplot[color = blue,
fill  = blue,
mark  = *,
only marks]
table[x = x, y = y] {\rwnormalmeanzero};

\addplot[color = blue,
fill  = blue,
mark  = x,
only marks]
table[x = x, y = y] {\rwhitenoise};

\addplot[color = orange,
fill  = orange,
mark  = x,
only marks]
table[x = x, y = z] {\rwhitenoise};

\addplot[color = orange,
fill  = orange,
mark  = *,
only marks]
table[x = x, y = z] {\rwnormalmeanzero};

\legend{Random Walk, Randomness} 
\end{axis}
\end{tikzpicture}

\caption{For randomness and a simulated symmetric normal random walk of length $N$, we compute $g_n$ for (a) $n = 4$ in blue and (b) $n = 5$ in orange.}
 \label{figure:mm}
\end{figure}

\section{KL Divergence Method}
\label{sec:KLsetup}

As described in Section 1.2, it is natural to interpret the permutation entropy of a time series as a measure of the divergence of the distribution of ordinal patterns from the uniform distribution as in white noise. Here, we compute the KL divergence to the distribution of patterns determined by a random walk model as a more nuanced measure of complexity. This is particularly relevant for data that is expected to be generated from a random walk process, such as stock closing prices. We also consider some periodic weather and heart rate data whose behavior lies in between these extremes. 

 This measure more accurately reflects the underlying process that generates our data. This is important as it allows us to more accurately explain the behavior of the time series. Additionally, observed deviations from the model are more meaningful in this setting since the random walk is chosen as a purposeful null model, rather than occurring as an artifact, as in the case of permutation entropy.

In the remainder of this paper, we construct null models by sampling from the distribution of observed steps from the data as described below. This approach has two advantages, first, we need not artificially select a particular inferential framework and second, it allows us to control for variance by generating many samples and comparing them to the observed data. Differences between the models and the empirical time series are then related to correlation between the steps. 
 
To determine how the behavior these data sets deviate from a random walk, we compute the relative frequency $p_{\pi}$ of each of the patterns $\pi$ of length four in the daily closing values, $X$.  Next, we construct a random walk model, $Z$ of length $M \gg N$, whose steps are taken by drawing at random from from the distribution of steps $\{X_2 - X_1, X_3 - X_2, \ldots, X_n - X_{n-1}\}$ in the original time series; we refer to $Z$ as the \textit{random walk associated with} $X$.  For each time series, we determine the deviation  from the model by computing 
$$\DKL_n(X) := \DKL_n(X||Z) =  \frac{1}{\log(n!)} \sum_{ \pi \in \mathcal{S}_n} p_{\pi} \log\left(\frac{p_{\pi}}{q_{\pi}}\right),$$
where $p_{\pi}$ is the relative frequency of $\pi$ in $X$ and $q_{\pi}$ is the relative frequency in $\pi$ in $Z$.  

In order to directly compare our results to permutation entropy we computed $1-\PE_n$ and $\DKL_n(X)$ for each of the data sets RAND, HEART, MEX, NYC, SP500, GE, and NORM RW. The results are displayed in Figure 7. The permutation entropy is plotted on the left and the random walk KL on the right. 

\begin{figure}[!h]
\begin{tikzpicture}[scale = .8]
\begin{axis} [ybar, bar width = 8pt,
    ymin=0,
    ylabel = $1-\PE_n$, 
    scaled y ticks=false,
      ytick={.025, .05, .075, .1,.125,.15},
    yticklabels={$.025$, $.05$,$ .075$, $.1$,$.125$,$.15$ }, 
    symbolic x coords={RAND, HEART,  MEX, NYC,  SP500, GE, NORM RW},
    xticklabel style={anchor=east, rotate= 90},
 width = 11 cm, 
 height = 6 cm, 
    xtick=data
]
\addplot[ybar, fill=blue]
table[x = x, y = y] {\petwothousand};
\addplot[ybar, fill=orange] 
  plot [error bars/.cd, y dir=both, y explicit]
  table [x= x, y = z] {\petwothousand};
\end{axis} 
\end{tikzpicture}
\begin{tikzpicture}[scale = .8]
\begin{axis} [ybar, bar width = 8pt,
scaled y ticks=false,
        ylabel = $\DKL_n$,
         ytick={.01,.02,.03, .04, .05},
     yticklabels={$.01$,$.02$,$.03$, $.04$, $.05$},   
    ymin=0,
    symbolic x coords={ RAND, HEART, MEX, NYC,  SP500, GE, NORM RW},
    xticklabel style={anchor=east, rotate= 90},
 width = 11 cm, 
 height = 6 cm, 
    xtick=data
]
\addplot[ybar, fill=blue]
table[x = x, y = z] {\kldivtwothousand};
\addplot[ybar, fill=gray!50] 
  plot [error bars/.cd, y dir=both, y explicit]
  table [x= x, y =y,  y error plus=yerror, y error minus= yerror] {\kldivtwothousand};
\end{axis} 
\end{tikzpicture}
\caption{ On the left, we compute $\PE_n$ for the time-series for $n = 4$ (in blue) and $n = 5$ (in orange). On the right, we compute $\DKL_n$ for $n = 4$ and the data of length $N = 2000$ (blue).  We generate $400$ random walks $\widehat X$ of length $N = 2000$ and compute $\DKL_n(\widehat X)$ for each. The mean and errors are plotted in gray. }
\label{fig:PEDKL2000}
\end{figure}
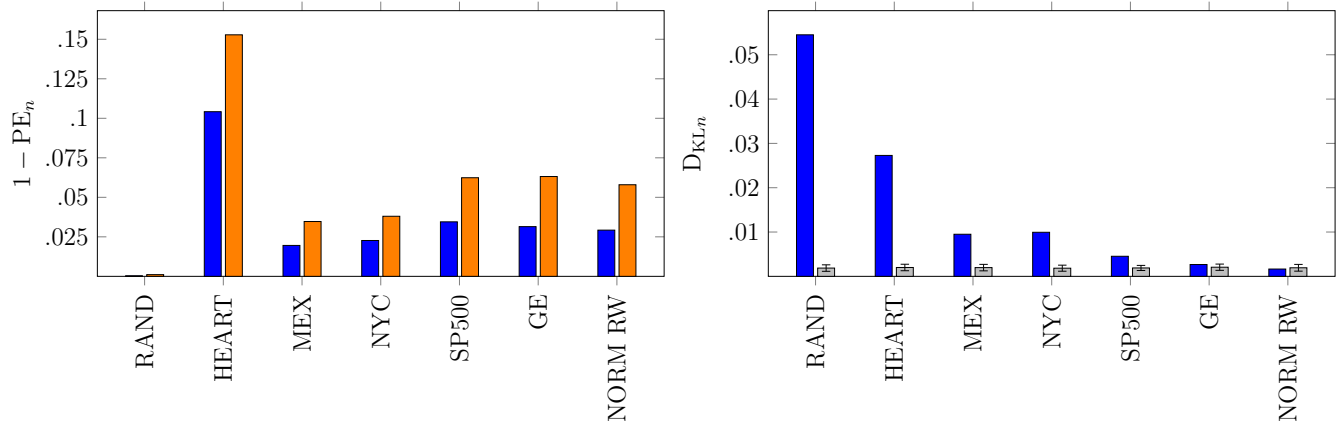

As shown in Figure \ref{fig:PEDKL2000} (a), the changes in heart rate are more correlated than steps in the other time series investigated here. However, when considering the KL divergence method, randomness is the time series furthest from a random walk.  This supports our view that the KL divergence method is frequently a better measure of deviation from a random walk than the permutation entropy of steps. The weather data sets are an interesting example where the structure is periodic and hence neither uniformly random or a fixed random walk. Thus, we see moderate performance under both measures. However, notice that $\PE_n$ only slightly distinguishes temperature data and a simulated random walk but the $D_{KLn}$ measure clearly separates them. 
 
 To add context to the value of the KL divergence, we simulated 400 random walks, $\widehat{X}$, associated with $X$ of length $N$ and calculated $\DKL_n(\widehat X)$ for each.   Using these simulations, we calculate the mean and standard deviation of the KL divergence of the simulated random walks against the model. These are plotted in Figure 7(b). Notice that the stock data is much better approximated by the random walk of its steps than any of the other time series. 
  
Finally, in order to determine how the length of the time series affects $\DKL$, we simulate a uniform random walk with $\mu =0$, $\widehat X$, of length $N$ and compare it to the distribution of patterns in the random walk.  The results mirror those of Figure 6. For $\widehat X$ of length $N = 1000$, $\DKL_4(\widehat X) \approx 0.10$ and $\DKL_5(\widehat X) \approx 0.11$.  For $\widehat X$ of length $N = 5000$, $\DKL_4(\widehat X) \approx 0.07$, where it remains for larger $N$, and $\DKL_5(\widehat X) \approx 0.01$, and falling to $0.007$ when $N=10,000$. This is expected behavior as the value goes to zero in the limit in expectation.

  Expanding on our remarks from the previous section, permutation entropy has frequently been used to study financial time series.  For instance, permutation entropy and the number of forbidden patterns for both closing values and returns were suggested as methods for distinguishing developed and emerging markets with the aim of using these measures to quantify stock market inefficiency \cite{zunino_forbidden_2009}. In this analysis, permutation entropy of returns were correlated with either being a developed or an emerging market, with emerging markets having smaller permutation entropy (i.e. more correlation). We plot these values for our data sets below and perform a more direct comparison in the following section.

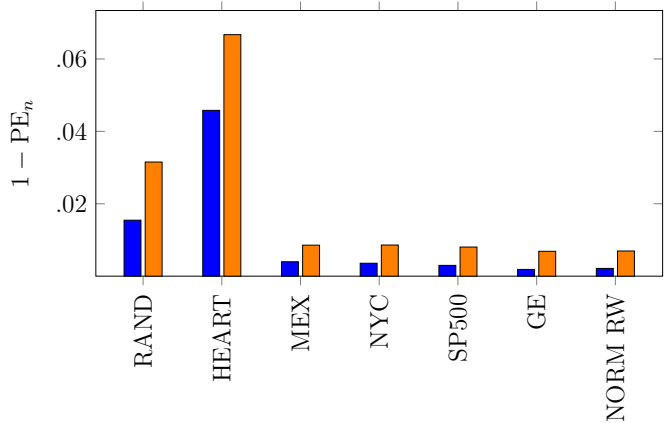
\begin{figure}[!h]
\begin{tikzpicture}[scale = .8]
\begin{axis} [ybar, bar width = 8pt,
    ymin=0,
    ylabel = $1-\PE_n$, 
    scaled y ticks=false,
     ytick={.02, .04, .06},
     yticklabels={$.02$, $.04$, $.06$},   
    symbolic x coords={RAND, HEART,  MEX, NYC,  SP500, GE, NORM RW},
    xticklabel style={anchor=east, rotate= 90},
 width = 11 cm, 
 height = 6 cm, 
    xtick=data
]
\addplot[ybar, fill=blue]
table[x = x, y = z] {\peofstepstwothousand};
\addplot[ybar, fill=orange] 
  plot [error bars/.cd, y dir=both, y explicit]
  table [x= x, y = zz] {\peofstepstwothousand};
\end{axis} 
\end{tikzpicture}
\caption{On the left, we compute $\PE_n$ for the time-series of steps $X_{n}-X_{n-1}$ for $n = 4$ (in blue) and $n = 5$ (in orange). This can be used as a measure of step independence and was presented in \cite{lim_rapid_2014} as a measure of volatility in developing economic markets.  }
\end{figure}

A careful analysis of this method demonstrates some key features that lead us to prefer the explicit random walk model. First, we note that this measure assigns very low values to all of the stock data. As we will see in the next section, this property limits the amount of information that can be extracted. Secondly, we note that the measure does not clearly distinguish the periodic weather data from the random walks.  Finally, the the permutation entropy of the steps discovers a relatively high value for i.i.d. randomly drawn data points because the difference between the random variables is not independent.

\section{Inefficiency in Financial Markets}\label{sec:KLstocks}

In this final section, we analyze the stock market data more closely, using the KL method and measure of momentum introduced above. Economic heuristics suggests that the most appropriate model of the stock market is that of the random walk, see for example \cite{fama_behavior_1965}.  Moreover, since a market whose prices are modeled by a random walk is considered efficient, the divergence of a market from that of a random walk serves as a measure of inefficiency \cite{fama_behavior_1965}.  Developing meaningful measures of market inefficiency is an important and well--studied question in finance.  Applying our method from Section \ref{sec:KLsetup} to a variety of stocks, we posit that a measure of inefficiency using the KL divergence from a random walk null model is preferable to the permutation entropy of returns.

First, using the measures
$$\varepsilon_{12\ldots(n{-}1)n} := p_{12\ldots (n{-}1)n} - (p_{12})^{n-1} \text{ and } \varepsilon_{n(n{-}1) \ldots 21} := p_{n(n{-}1)\ldots21} - (p_{21})^{n-1}$$
that we developed in Section \ref{sec:MM}, we capture the momentum phenomena observed in financial markets. Indeed, $\varepsilon_{12\ldots(n{-}1)n} > 0$ suggests a presence of upward momentum and  $\varepsilon_{n(n{-}1)\ldots21} > 0$ suggests a presence of downward momentum.  As depicted in Figure \ref{fig:momentum}, for each of the stocks considered, the values of $\varepsilon_{\pi}$ for $\pi = 1234$ and $\pi = 4321$ are positive, suggesting a presence of both upward and downward momentum in these markets.  
Both of these results accord with economic data reported by the NBER \cite{jegadeesh_returns_1993, jegadeesh_profitability_1999}  

\begin{figure}
\begin{tikzpicture}

\begin{axis} 
[
	ybar, bar width = 8pt,
    ymin=-0.005,
    scaled y ticks=false,
         ytick={0, .01, .02},
     yticklabels={$0$, $.01$, $.02$},   
    symbolic x coords={SP500, GE, AAPL, BAC, KO, UPS, AMZN},
    xticklabel style={anchor=east, rotate= 90},
 	width = 11 cm,  	  
 	height = 6 cm, 
 	ylabel = $\varepsilon_{\pi}$,
    xtick=data
]
		\addplot[ fill = blue] table[x=x,y=elfour] {\momentum};
					\addplot[fill = red] table[x=x,y=ehfour] {\momentum};			
						
\end{axis} 
\draw (0,.78) -- (9.401,.78);
\end{tikzpicture}
\caption{Values of $\varepsilon_{\pi}$ for $\pi = 1234$ in blue, and $\pi = 4321$ in red. Larger values of $\varepsilon_{\pi}$ correspond to markets containing longer increasing (resp. decreasing) runs than predicted by the associated random walk model. }  
\label{fig:momentum}
\end{figure}
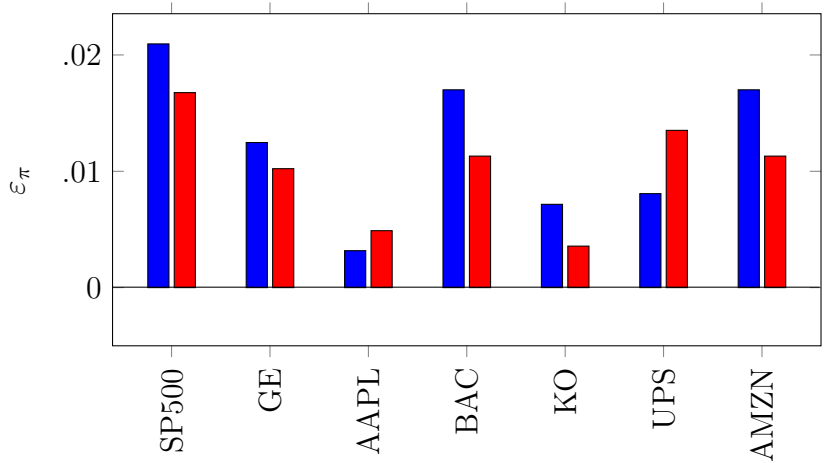

Although the previous result suggests that a random walk may not capture all of the information about the stock behavior since the momentum is a measure of correlation of the steps, which we have assumed to be i.i.d., we conclude with two examples demonstrating the advantages of the random walk divergence over permutation entropy. For each of the stocks under consideration, we form 400 random walks, $\widehat X$, associated to $X$ of length $N =3777$ (the length of $X$).  Then, to determine the significance of $\DKL_4(X)$, we compute $\DKL_4(\widehat X)$ for each.

The results of this experiment are presented in Figure 9.
As depicted in Figure \ref{fig:PEDKLstocks}, Apple stock (AAPL) was furthest from a random walk, perhaps a result of calendar year phenomena associated the release of new products.  On the other hand, large industrial stocks such as General Electric, Coke, and United Parcel Service (resp. GE, KO, and UPS) adhere more closely to the random walk model and are considered more efficient markets in this analysis.

\begin{figure}[!h]
\begin{tikzpicture}[scale = .8]
\begin{axis} [ybar, bar width = 8pt,
    ymin=0,
     ylabel = $1-\PE_n$,
    symbolic x coords={SP500, AAPL, AMZN, BAC, GE, KO, UPS},
    xticklabel style={anchor=east, rotate= 90},
 width = 11 cm, 
 height = 6 cm, 
    xtick=data
]
\addplot[ybar, fill=blue]
table[x = x, y = z] {\peofstepsstocks};
\addplot[ybar, fill=orange] 
  plot [error bars/.cd, y dir=both, y explicit]
  table [x= x, y = zz] {\peofstepsstocks};
\end{axis} 
\end{tikzpicture}
\begin{tikzpicture}[scale = .8]
\begin{axis} [ybar, bar width = 8pt,
    ymin=0,
     ylabel = $\DKL_n$,
    symbolic x coords={SP500, AAPL, AMZN, BAC, GE, KO, UPS},
    xticklabel style={anchor=east, rotate= 90},
 width = 11 cm, 
 height = 6 cm, 
    xtick=data
]
\addplot[ybar, fill=blue]
table[x = x, y = z] {\kldivstocks};
\addplot[ybar, fill=gray!50] 
  plot [error bars/.cd, y dir=both, y explicit]
  table [x= x, y =y,  y error plus=yerror, y error minus= yerror] {\kldivstocks};
\end{axis} 
\end{tikzpicture}

\caption{On the left, we compute $\PE_n$ for the time series of steps for $n = 4$ (in blue) and $n = 5$ (in orange).  On the right, we compute $\DKL_n$ for $n = 4$ and the data of length $N = 2000$ (blue).  We generate $400$ random walks $\widehat X$ (associated with $X$) of length $N = 2000$ and compute $\DKL_n(\widehat X)$ for each. The mean and errors are plotted in gray. }
\label{fig:PEDKLstocks}
\end{figure}
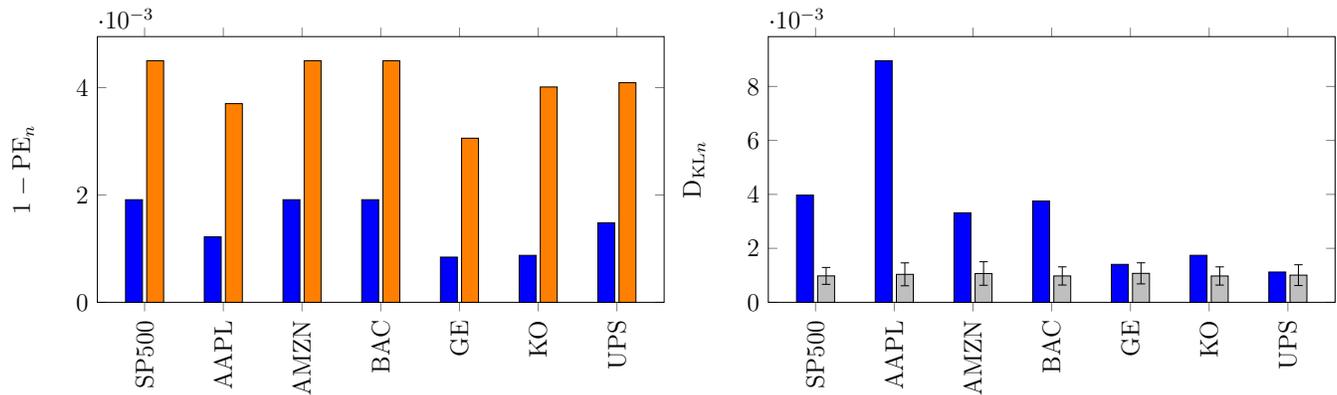

As a final application of these methods, we use historical S\&P500 closing prices from January 1958 until January 2017 and plot our measure of inefficiency, $\DKL_4$, over time, comparing to the permutation entropy of the steps.  For each year from 1960 until 2014, for the five year range surrounding the year (i.e. from January 1 of two years prior to December 31 of two years after, $N \approx 1258$), we compute $\DKL_4$ for the S\&P500, see Figure \ref{fig:KLLongSP}.

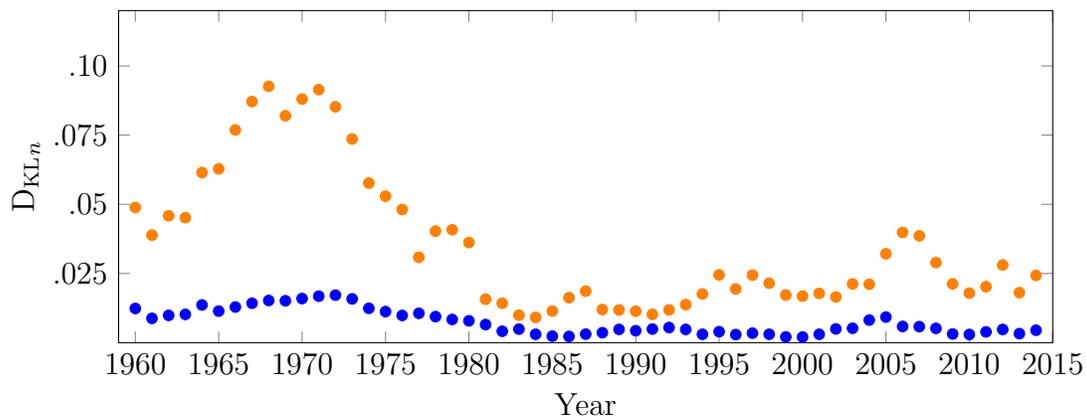
\begin{figure}
\begin{tikzpicture}[scale = 1]
\begin{axis}[
 width = 14 cm, 
 height = 6 cm, 
xlabel = Year,
   scaled y ticks=false,
         ytick={.025,.05,.075,.1},
     yticklabels={$.025$, $.05$, $.075$, $.10$},   
x tick label style={/pgf/number format/.cd,fixed,precision=3, set thousands separator={}},
xmax   = 2015,
xmin   = 1959,
ylabel = $\DKL_n$, 
ymax   = .12,
ymin   = 0
]
\addplot[color = orange,
fill  = orange,
mark  = *,
only marks]
table[x = x, y = KL] {\longrunSP};  


\addplot[color = blue,
fill  = blue,
mark  = *,
only marks]
table[x = x, y = PEofsteps] {\longrunSP};
\end{axis}
\end{tikzpicture}
\caption{Computation of $\DKL_4$ (orange) and the permutation entropy of the steps (blue) on historical S\&P500 daily closing prices during each 5 year window surrounding the year on the $x$-axis.  Both of these metrics  can be treated as a proxy for inefficiency but the $\DKL_4$ provides significantly more information. }
\label{fig:KLLongSP}
\end{figure}

 The general trends depicted in the plot of $\DKL_4$ resonate with the evolution of technology and economic events of that time, while the permutation entropy of the steps is less informative. In particular, we can see the decline in inefficiency as a result of computerized trading, as well as the stock market crash of 1989, the 2000 technology bubble, and the 2008 financial crisis causing an increase in variability and distance from the model. The results presented here are similar to those in \cite{hou_characterizing_2017} for the Shanghai and  Shenzhen Stock Exchanges. 

\section{Conclusion}

In order to account for observed behavior of the distribution of ordinal patterns in time series from economics and other fields, we have introduced a measure of complexity based on random walk null models. Since much of the structure of the ordinal patterns appearing in these financial time series is explained by the underlying process of a random walk,  this measure is better suited for such time series than previous methods based on permutation entropy. 
We provided theoretical and numerical results on the distribution of patterns in the context of random walk models and provided a set of tools for analyzing the complexity of data modeled by time series. Additionally, we have applied our methods to examples from several different domains  in order to validate their usefulness. Not all time series data plausibly arises from random walk processes but for those that do the methods presented in this paper provide a principled method for studying their complexity and inefficiency.

\section*{Acknowledgements} The authors would like to express their gratitude to Peter Doyle, Sergi Elizalde, Nishant Malik, Megan Martinez, Scott Pauls and  Dan Rockmore  for their helpful comments and guidance. 

\bibliographystyle{plain}
\bibliography{Time_Series}

\newpage

\appendix
\section{Data Plots}

In this appendix we display plots of the data sets described in Section 1.5 and used throughout the paper. Plots generated with Mathematica \cite{Mathematica}.

\begin{figure}[!h]
\subfloat[RAND]{\includegraphics[scale=.33]{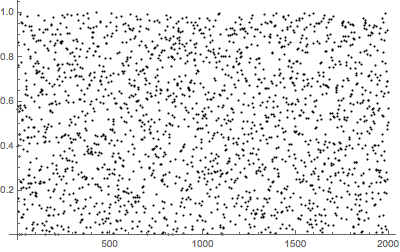}}\quad
\subfloat[NORM RW]{\includegraphics[scale=.33]{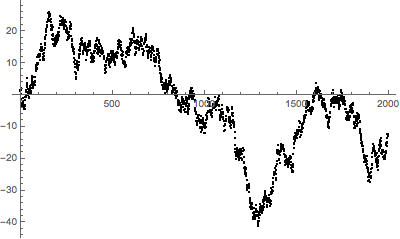}}\quad \quad
\subfloat[UNIF RW]{\includegraphics[scale=.33]{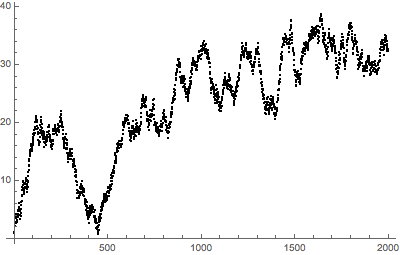}}\\ 
\subfloat[DRIFT RW]{\includegraphics[scale=.33]{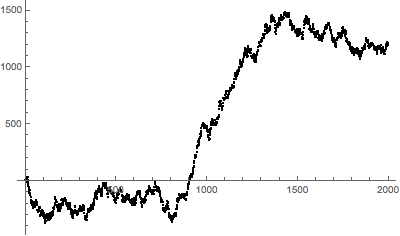}}\quad \quad
\subfloat[MEX]{\includegraphics[scale=.33]{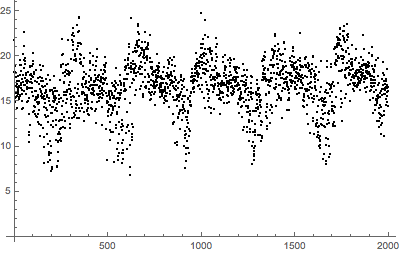}}\quad \quad 
\subfloat[NYC]{\includegraphics[scale=.33]{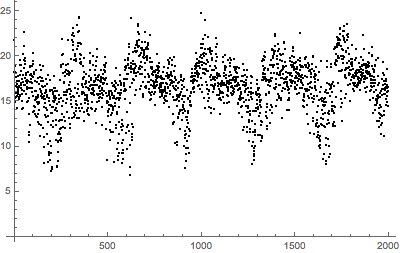}} \\
\subfloat[GE]{\includegraphics[scale=.33]{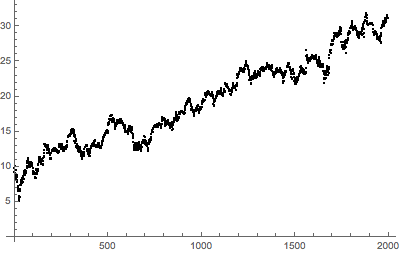}} \quad \quad
\subfloat[HEART]{\includegraphics[scale=.33]{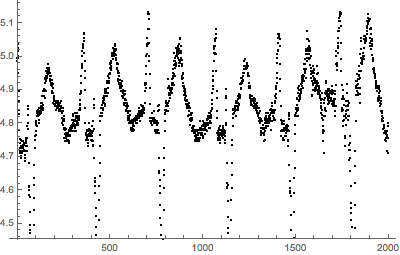}} \quad \quad
\subfloat[STOCK SP500]{\includegraphics[scale=.33]{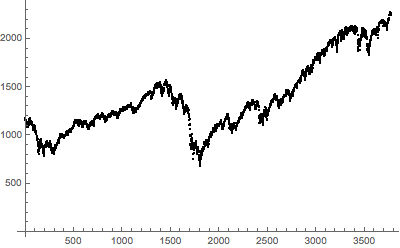}}\\ 
\subfloat[STOCK AAPL]{\includegraphics[scale=.33]{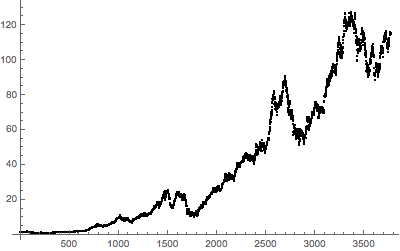}} \quad \quad 
\subfloat[STOCK AMZN]{\includegraphics[scale=.33]{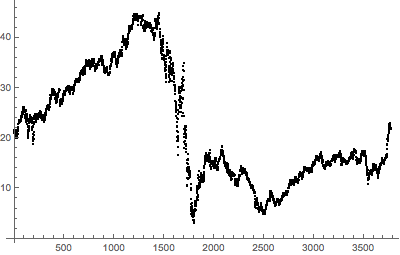}} \quad \quad
\subfloat[STOCK BAC]{\includegraphics[scale=.33]{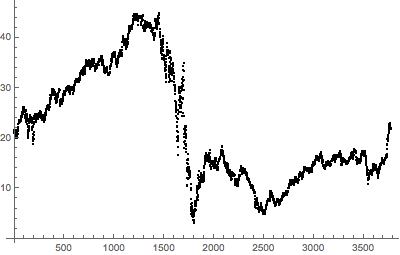}}\\ 
\subfloat[STOCK GE]{\includegraphics[scale=.33]{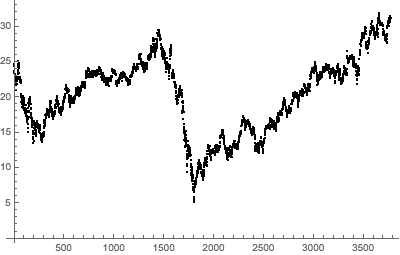}} \quad \quad
\subfloat[STOCK KO]{\includegraphics[scale=.33]{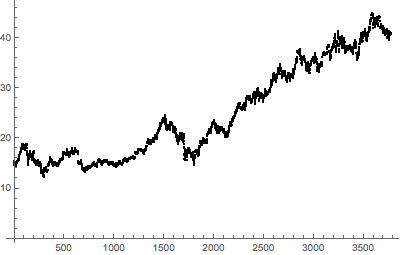}} \quad \quad
\subfloat[STOCK UPS]{\includegraphics[scale=.33]{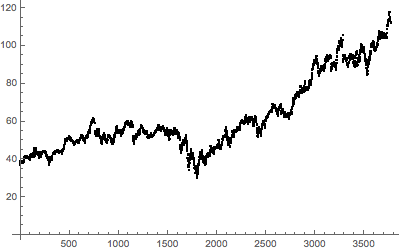}}  
\caption{Graphs of the time-series used throughout this paper, see Section 1.5. Time series (A)-(H) are of length $N = 2000$.  Stock data (I)-(O), used in Section 5,  are closing prices for trading days from January 1, 2002 to January 1, 2017 and of length $N = 3777$.  }
\end{figure}

\newpage

\section{Null Model Distributions}
Here we give the expected distributions of ordinal patterns for the uniform and normal random walk models as determined by Proposition 5. Recall that for the normal distribution the values do not depend on the variance when the mean is zero. For the uniform distribution the $\mu=0$ case is equivalent to setting $b=\frac12$. 

\begin{table}[h!]
\begin{tabular}{|c| c | c | c| }
\hline 
Pattern & Normal: $\mu = 0$ & Uniform: $\mu = 0$ & Uniform: $\P(Y > 0) = b$  \\
\hline
\{123\} & $1/4$ & $1/4$ & $b^2$ \\
\hline 
\{132, 213\} & $1/8$ & $1/8$  & $(1/2)(1-b)^2$\\
\hline  
\{213, 132\} & $1/8$ & $1/8$ & $(1/2)(b^2 + 2b - 1)$ \\
\hline  
\{321\} & $1/4$ & $1/4$  & $(1-b)^2$ \\
\hline
\hline 
\{1234\} & 0.1250 & $1/8$ & $b^3$ \\
\hline 
\{1243, 2134\} & 0.0625 & 1/16 & $(1/2)b (1-b)  (3b-1)$\\
\hline 
\{1324\} & 0.0417 & 1/24 & $(1/3) (1-b) (7 b^2 -5 b+1)$ \\
\hline
\{1342, 3124\} & 0.0208& 1/24 & $(1/6) (1-b)^2 (4 b-1)$ \\
\hline 
\{1423, 2314\} & 0.0355 & 1/48  & $(1/6) (1-b)^2 (5 b-2)$ \\
\hline
\{1432, 2143, 3214\} & 0.0270 & 1/48 & $\begin{cases}  (1/6)(2-24 b+48 b^2-15 b^3)&	\text{if }  b \leq 2/3 \\
(b-1)^2 (2b-1) & 	\text{if } b > 2/3 \end{cases} $\\
\hline
\{2341, 3412, 4123\} & 0.0270 & 1/48 & $(1/6) (1-b)^3$ \\
\hline
\{2413\} &0.0146  &1/48  & $(1/6) (1-b)^3$ \\
\hline
\{2431, 4213\} & 0.0208 & 1/24 &$\begin{cases} (1/6) (24 b^3 -45 b^2 +27 b-5)	& \text{ if } b \leq 2/3 \\ 
(1/2) (1-b)^3	&  \text{ if } b>2/3 
 \end{cases} $\\
\hline
\{3142\} & 0.0146 & 1/48 & $\begin{cases} (1/6)(25 b^3-48 b^2 +30b -6)	& \text{ if } b \leq 2/3 \\  (1/3) (1-b)^3  & \text{ if } b > 2/3
\end{cases}$\\
\hline
\{3241, 4132\} & 0.0355 & 1/48 & $(1/6) (1-b)^3$ \\
\hline
\{3421, 4312\} &  0.0625 & 1/16 &$(1/2) (1-b)^3$\\
\hline
\{4231\} &  0.0417 & 1/24 & $(1/3) (1-b)^3$\\
\hline
\{4321\}  & 0.1250 & 1/8  & $(1-b)^3$ \\
\hline
\end{tabular}
\vspace{1em}

\caption{The values of $\P_Z(\pi)$ for the normal distribution with $\mu = 0$ and in the uniform case for $\P(Y>0) = b$, where $\frac{1}{2} \leq b \leq 1$.}
\label{tab:normunifprobabilites}
\end{table}
 
\newpage

\section{Permutation Equivalence Classes} In \cite{martinez_frequency_2015,MeganThesis}, the authors defined equivalence relation on permutations by $\pi \sim \tau$ if $\P_Z(\pi) = \P_Z(\tau)$ for any random walk, $Z$.  They show that $\pi \sim \tau$ if the permutations can be related by a sequence of combinatorial moves. For completeness, we list the equivalence classes described by their result. Although the existence of these classes was categorized theoretically this is the first time they have been explicitly computed \cite{martinez_frequency_2015}.  We use the classes to define the function $g_n$ in Section 3.

 For $n = 3$, the classes are
$$\{123\}, \{132, 213\}, \{213, 132\}, \{321\}.$$

For $n = 4$, the classes are
$$\{1234\}, \{1243, 2134\}, \{1324\}, \{1342, 3124\}, \{1423, 2314\}, \{1432, 2143, 3214\}, $$ $$ \{2341, 3412, 4123\}, \{2413\}, \{2431, 4213\}, \{3142\}, \{3241, 4132\}, \{3421, 4312\}, \{4231\}, \{4321\}.$$

For $n = 5$, the classes are 
$$ \{12345\}, \{14325\}, \{21354\}, \{21453\}, \{25314\}, \{41352\}, \{45312\},\{52341\}, \{54321\},$$ $$  \{12543, 32145\}, \{13245, 12435\}, \{13425, 14235\},\\  \{15243, 32415\}, \{15342, 42315\}, \{15432, 43215\},$$ $$ \{21345, 12354\}, \{21435, 13254\}, \{21543, 32154\}, \{23145, 12534\}, \\ \{23415, 15234\}, \{24153, 31524\},$$ $$ \{24315, 15324\}, \{24513, 35124\}, \{24531, 53124\}, \{25134, 23514\}, \{25341, 52314\}, \\  \{25413, 35214\}, $$ $$\{25431, 53214\}, \{31245, 12453\}, \{31425, 14253\}, \{31542, 42153\}, \{32514, 25143\}, \{32541, 52143\}, $$ $$\\ \{35142, 42513\}, \{35241, 52413\}, \{41235, 13452\}, \\ \{41253, 31452\}, \{41325, 14352\}, \{41523, 34152\},$$ $$ \{41532, 43152\}, \{42135, 13542\}, \{42351, 51342\}, \\ \{45231, 53412\}, \{51324, 24351\}, \{51423, 34251\},$$ $$ \{51432, 43251\}, \\ \{53142, 42531\}, \{53241, 52431\}, \{54123, 34521\}, \{54132, 43521\}, \{54213, 35421 \},$$ $$ \{54231, 53421\}, \\ \{54312, 45321\}, \{31254\},  \{43125, 14532\}, \{34125, 14523\},  \{13524, 24135\},$$ $$ \\ \{35412, 52134, 45213, 23541\},
  \{51243, 32451\} \{43512, 45132\}, \{23451, 45123, 34512, 51234\},$$ $$
  \{21534, 23154,  15423, 34215\}.$$

\end{document}